\documentclass[sigconf,natbib]{acmart}

\copyrightyear{2021}
\acmYear{2021}
\setcopyright{acmlicensed}\acmConference[WSDM '21]{Proceedings of the Fourteenth ACM International Conference on Web Search and Data Mining}{March 8--12, 2021}{Virtual Event, Israel}
\acmBooktitle{Proceedings of the Fourteenth ACM International Conference on Web Search and Data Mining (WSDM '21), March 8--12, 2021, Virtual Event, Israel}
\acmPrice{15.00}
\acmDOI{10.1145/3437963.3441794}
\acmISBN{978-1-4503-8297-7/21/03}

\PassOptionsToPackage{dvipsnames}{xcolor}

\usepackage{dsfont}
\usepackage{acronym}
\usepackage{amsmath}
\usepackage[noend]{algorithmic}
\usepackage{algorithm}
\usepackage{bbm}
\usepackage{beramono}
\usepackage{bm}
\usepackage{booktabs}
\usepackage[skip=0pt]{caption}
\usepackage{cleveref}
\usepackage[T1]{fontenc}
\usepackage{graphicx}
\usepackage{hyperref}
\usepackage{listings}
\usepackage{multirow}
\usepackage{natbib}
\usepackage{subcaption}
\usepackage{tabularx}
\usepackage[dvipsnames]{xcolor}
\usepackage{enumerate}
\usepackage[inline]{enumitem}
\usepackage{numprint}
\usepackage[normalem]{ulem}
\usepackage{mleftright}

\copyrightyear{2021}
\acmYear{2021}
\setcopyright{rightsretained}
\acmConference[WSDM '21]{Proceedings of the 14th ACM International Conference on Web Search and Data Mining}{March 8--12, 2021}{Jerusalem, Israel}
\acmBooktitle{The 14th ACM International Conference on Web Search and Data Mining (WSDM '21), March 8--12, 2021, Jerusalem, Israel}
\acmPrice{}
\definecolor{rj}{RGB}{0, 150, 0}
\definecolor{mdr}{RGB}{200, 0, 0}
\definecolor{ho}{RGB}{0, 50, 150}

\acrodef{IR}{Information Retrieval}
\acrodef{LTR}{Learning to Rank}
\acrodef{ARP}{Average Relevance Position}
\acrodef{DCG}{Discounted Cumulative Gain}
\acrodef{EM}{Expectation Maximization}
\acrodef{CTR}{Click-Through-Rate}
\acrodef{IPS}{Inverse-Propensity-Scoring}
\acrodef{LogOpt}{Logging-Policy Optimization Algorithm}
\acrodef{RCTR}{Relevant-Click-Through-Rate}
\acrodef{DBGD}{Dueling Bandit Gradient Descent}
\acrodef{COLTR}{Counterfactual Online Learning to Rank}
\acrodef{PDGD}{Pairwise Differentiable Gradient Descent}
\acrodef{NRCTR}{Normalized RCTR}
\acrodef{NDCG}{Normalized DCG}

\allowdisplaybreaks

\author{Harrie Oosterhuis}
\affiliation{%
	\institution{Radboud University}
	\city{Nijmegen}
	\country{The Netherlands}
}
\email{harrie.oosterhuis@ru.nl}
 
\author{Maarten de Rijke}
\orcid{0000-0002-1086-0202}
\affiliation{
 \institution{University of Amsterdam \& Ahold Delhaize}
 \city{Amsterdam}
 \country{The Netherlands}
}
\email{derijke@uva.nl}

\acmSubmissionID{38}

\title{Unifying Online and Counterfactual Learning to Rank}
\subtitle{A Novel Counterfactual Estimator that Effectively Utilizes Online Interventions}

\allowdisplaybreaks

\begin{document}

\begin{abstract}
Optimizing ranking systems based on user interactions is a well-studied problem.
State-of-the-art methods for optimizing ranking systems based on user interactions are divided into online approaches -- that learn by directly interacting with users -- and counterfactual approaches -- that learn from historical interactions.
Existing online methods are hindered without online interventions and thus should not be applied counterfactually.
Conversely, counterfactual methods cannot directly benefit from online interventions.

We propose a novel intervention-aware estimator for both counterfactual and online \ac{LTR}.
With the introduction of the intervention-aware estimator, we aim to bridge the online/counterfactual \ac{LTR} division as it is shown to be highly effective in both online and counterfactual scenarios.
The estimator corrects for the effect of position bias, trust bias, and item-selection bias by using corrections based on the behavior of the logging policy and on online interventions: changes to the logging policy made during the gathering of click data.
Our experimental results, conducted in a semi-synthetic experimental setup, show that, unlike existing counterfactual \ac{LTR} methods, the intervention-aware estimator can greatly benefit from online interventions.
\end{abstract}

\maketitle

\acresetall

\section{Introduction}
\label{sec:intro}

Ranking systems form the basis for most search and recommendation applications~\citep{liu2009learning}.
As a result, the quality of such systems can greatly impact the user experience, thus it is important that the underlying ranking models perform well.
The \ac{LTR} field considers methods to optimize ranking models. 
Traditionally this was based on expert annotations.
Over the years the limitations of expert annotations have become apparent; some of the most important ones are:
\begin{enumerate*}[label=(\roman*)]
\item they are expensive and time-consuming to acquire~\citep{qin2013introducing, Chapelle2011};
\item in privacy-sensitive settings expert annotation is unethical, e.g., in email or private document search~\citep{wang2018position}; and
\item often expert annotations appear to disagree with actual user preferences~\citep{sanderson2010}.
\end{enumerate*}

User interaction data solves some of the problems with expert annotations:
\begin{enumerate*}[label=(\roman*)]
\item interaction data is virtually free for systems with active users;
\item it does not require experts to look at potentially privacy-sensitive content;
\item interaction data is indicative of users' preferences.
\end{enumerate*}
For these reasons, interest in \ac{LTR} methods that learn from user interactions has increased in recent years.
However, user interactions are a form of implicit feedback and generally also affected by other factors than user preference~\citep{joachims2017accurately}.
Therefore, to be able to reliably learn from interaction data, the effect of factors other than preference has to be corrected for.
In clicks on rankings three prevalent factors are well known:
\begin{enumerate*}[label=(\roman*)]
\item \emph{position bias}: users are less likely to examine, and thus click, lower ranked items~\citep{craswell2008experimental};
\item \emph{item-selection bias}: users cannot click on items that are not displayed~\citep{ovaisi2020correcting, oosterhuis2020topkrankings}; and
\item \emph{trust bias}: because users trust the ranking system, they are more likely to click on highly ranked items that they do not actually prefer~\citep{agarwal2019addressing, joachims2017accurately}.
\end{enumerate*}
As a result of these biases, which ranking system was used to gather clicks can have a substantial impact on the clicks that will be observed.
Current \ac{LTR} methods that learn from clicks can be divided into two families:
\emph{counterfactual approaches}~\citep{joachims2017unbiased} -- that learn from historical data, i.e., clicks that have been logged in the past -- and \emph{online approaches}~\citep{yue2009interactively} -- that can perform interventions, i.e., they can decide what rankings will be shown to users.
Recent work has noticed that some counterfactual methods can be applied as an online method~\citep{jagerman2019comparison}, or vice versa~\citep{zhuang2020counterfactual, ai2020unbiased}.
Nonetheless, every existing method was designed for either the online or counterfactual setting, never both.

In this work, we propose a novel estimator for both counterfactual and online \ac{LTR} from clicks: the \emph{intervention-aware estimator}.
The intervention-aware estimator builds on ideas that underlie the latest existing counterfactual methods: the policy-aware estimator~\citep{oosterhuis2020topkrankings} and the affine estimator~\citep{vardasbi2020trust}; and expands them to consider the effect of online interventions.
It does so by considering how the effect of bias is changed by an intervention, and utilizes these differences in its unbiased estimation.
As a result, the intervention-aware estimator is both effective when applied as a counterfactual method, i.e., when learning from historical data, and as an online method where online interventions lead to enormous increases in efficiency.
In our experimental results the intervention-aware estimator is shown to reach state-of-the-art \ac{LTR} performance in both online and counterfactual settings, and it is the only method that reaches top-performance in both settings.

The main contributions of this work are:
\begin{enumerate}[align=left,leftmargin=*]
\item A novel intervention-aware estimator that corrects for position bias, trust bias, item-selection bias, and the effect of online interventions.
\item An investigation into the effect of online interventions on state-of-the-art counterfactual and online \ac{LTR} methods.
\end{enumerate}

\section{Interactions with Rankings}
\label{sec:userinteractions}

The theory in this paper assumes that three forms of interaction bias occur: position bias, item-selection bias, and trust bias.

\emph{Position bias} occurs because users only click an item after examining it, and users are more likely to examine items displayed at higher ranks~\citep{craswell2008experimental}.
Thus the rank (a.k.a.\ position) at which an item is displayed heavily affects the probability of it being clicked.
We model this bias using $P(E = 1 \mid k)$: the probability that an item $d$ displayed at rank $k$ is examined by a user $E$~\citep{wang2018position}.

\emph{Item-selection bias} occurs when some items have a zero probability of being examined in some displayed rankings~\citep{ovaisi2020correcting}.
This can happen because not all items are displayed to the user, or if the ranked list is so long that no user ever considers the entire list.
We model this bias by stating:
$
\exists k, \forall k', \, (k' > k \rightarrow P(E = 1 \mid k') = 0),
$
i.e., there exists a rank $k$ such that items ranked lower than $k$ have no chance of being examined.
The distinction between position bias and item-selection bias is important because some methods can only correct for the former if the latter is not present~\citep{oosterhuis2020topkrankings}.

Finally, \emph{trust bias} occurs because users trust the ranking system and, consequently, are more likely to perceive top ranked items as relevant even when they are not~\citep{joachims2017accurately}.
We model this bias using: $P(C = 1 \mid k, R, E)$: the probability of a click conditioned on the displayed rank $k$, the relevance of the item $R$, and examination $E$.

To combine these three forms of bias into a single click model, we follow \citet{agarwal2019addressing} and write:
\begin{equation}
\begin{split}
P(C&=1 \mid d, k, q)
\\
&= P(E = 1 \mid k)\big(P(C = 1 \mid k, R = 0, E=1)P(R=0 \mid d, q)\\
 &\qquad\qquad\quad\,\, + P(C = 1 \mid k, R = 1, E=1)P(R=1 \mid d, q)\big),
\end{split}
\label{eq:click-probability}
\end{equation}
where $P(R=1 \mid d, q)$ is the probability that an item $d$ is deemed relevant w.r.t.\ query $q$ by the user.
An analysis on real-world interaction data performed by \citet{agarwal2019addressing}, showed that this model better captures click behavior than models that only capture position bias~\citep{wang2018position} on search services for retrieving cloud-stored files and emails.

To simplify the notation, we follow \citet{vardasbi2020trust} and adopt:
\begin{equation}
\begin{split}
\alpha_k &= P(E = 1 \mid k)\big(P(C = 1 \mid k, R = 1, E=1)
\\ & \phantom{= P(E = 1 \mid k)\big(P} -  P(C = 1 \mid k, R = 0, E=1)\big),
\\
\beta_k &= P(E = 1 \mid k)P(C = 1 \mid k, R = 0, E=1).
\end{split}
\end{equation}
This results in a compact notation for the click probability \eqref{eq:click-probability}:
\begin{equation}
P(C=1 \mid d, k, q) = \alpha_k P(R=1 \mid d, q) + \beta_k.
\label{eq:clickmodelused}
\end{equation}
For a single ranking $y$, let $k$ be the rank at which item $d$ is displayed in $y$; we denote $\alpha_k = \alpha_{d,y}$ and $\beta_k = \beta_{d,y}$.
This allows us to specify the click probability conditioned on a ranking $y$:
\begin{equation}
P(C=1 \mid d, y, q) = \alpha_{d,y} P(R=1 \mid d, q) + \beta_{d,y}.
\label{eq:trustbiasmodel}
\end{equation}
Finally, let $\pi$ be a ranking policy used for logging clicks, where $\pi(y \mid q)$ is the probability of $\pi$ displaying ranking $y$ for query $q$, then the click probability conditioned on $\pi$ is:
\begin{equation}
P(C=1 \mid d, \pi, q) = \sum_{y} \pi(y \mid  q) \mleft( \alpha_{d,y} P(R=1 \mid d, q) + \beta_{d,y} \mright).
\label{eq:clickpolicy}
\end{equation}
The proofs in the remainder of this paper will assume this model of click behavior.

\section{Background}

In this section we cover the basics on \ac{LTR} and counterfactual \ac{LTR}.

\subsection{Learning to Rank}
\label{sec:background:LTR}

The field of \ac{LTR} considers methods for optimizing ranking systems w.r.t. ranking metrics.
Most ranking metrics are additive w.r.t.\ documents; let $P(q)$ be the probability that a user-issued query is query $q$, then the metric reward $\mathcal{R}$ commonly has the form:
\begin{equation}
\mathcal{R}(\pi)
=
\sum_{q} P(q) \sum_{d \in D_q} \lambda(d \mid D_q, \pi, q) P(R = 1 \mid d, q).
\label{eq:truereward}
\end{equation}
Here, the $\lambda$ function scores each item $d$ depending on how $\pi$ ranks $d$ when given  the preselected item set $D_q$; $\lambda$ can be chosen to match a desired metric, for instance, the common \ac{DCG} metric~\citep{jarvelin2002cumulated}:
\begin{equation}
\lambda_{\text{DCG}}(d \mid D_q, \pi, q)
=
\sum_y \pi(y \mid q) \mleft( \log_2(\text{rank}(d \mid y) + 1) \mright)^{-1}.
\label{eq:dcglambda}
\end{equation}
Supervised \ac{LTR} methods can optimize $\pi$ to maximize $\mathcal{R}$ if relevances $P(R = 1 \,|\, d, q)$ are known~\citep{wang2018lambdaloss, liu2009learning}. However in practice, finding these relevance values is not straightforward.

\subsection{Counterfactual Learning to Rank}
\label{sec:background:counterfactual}

Over time, limitations of the supervised \ac{LTR} approach have become apparent. 
Most importantly, finding accurate relevance values $P(R = 1 \mid d, q)$ has proved to be impossible or infeasible in many practical situations~\citep{wang2016learning}.
As a solution, \ac{LTR} methods have been developed that learn from user interactions instead of relevance annotations.
Counterfactual \ac{LTR} concerns approaches that learn from historical interactions.
Let $\mathcal{D}$ be a set of collected interaction data over $T$ timesteps; for each timestep $t$ it contains the user issued query $q_t$, the logging policy $\bar{\pi}_t$ used to generate the displayed ranking $\bar{y}_t$, and the clicks $c_t$ received on the ranking:
\begin{equation}
\mathcal{D} = \{(\bar{\pi}_t, q_t, \bar{y}_t, c_t)\}_{t=1}^{T},
\end{equation}
where $c_t(d) \in \{0,1\}$ indicates whether item $d$ was clicked at timestep $t$.
While clicks are indicative of relevance they are also affected by several forms of bias, as discussed in Section~\ref{sec:userinteractions}.

Counterfactual \ac{LTR} methods utilize estimators that correct for bias to unbiasedly estimate the reward of a policy $\pi$.
The prevalent methods introduce a function $\hat{\Delta}$ that transforms a single click signal to correct for bias.
The general estimate of the reward is:
\begin{equation}
\hat{\mathcal{R}}(\pi \mid \mathcal{D}) =  \frac{1}{T} \sum_{t=1}^T \sum_{d \in D_{q_t}}  \lambda(d \mid D_{q_t}, \pi, q) \hat{\Delta}(d \mid \bar{\pi}_t, q_t, \bar{y}_t, c_t).
\label{eq:estimatedreward}
\end{equation}
We note the important distinction between the policy $\pi$ for which we estimate the reward, and the policy $\bar{\pi}_t$ that was used to gather interactions.
During optimization only $\pi$ is changed in order to maximize the estimated reward.

The original \ac{IPS} based estimator introduced by \citet{wang2016learning} and \citet{joachims2017unbiased} weights clicks according to examination probabilities:
\begin{equation}
\hat{\Delta}_{\text{IPS}}(d \mid \bar{y}_t, c_t) =  \frac{c_t(d)}{P(E = 1 \mid \bar{y}_t, d)}.
\end{equation}
This estimator results in unbiased optimization under two requirements.
First, every relevant item must have a non-zero examination probability in all displayed rankings:
\begin{equation}
\forall t,  \forall d \in D_{q_t}\, \left(P(R=1 \mid d, q_t) > 0
\rightarrow P(E = 1 \mid \bar{y}_t, d) > 0\right).
\end{equation}
Second, the click probability conditioned on relevance on examined items should be the same on every rank:
\begin{equation}
\forall k, k'\, \left(P(C \mid k, R, E=1) = P(C \mid k', R, E=1)\right),
\label{eq:notrustreq}
\end{equation}
i.e., no trust bias is present.
These requirements illustrate that this estimator can only correct for position bias, and is biased when item-selection bias or trust bias is present.
For a proof we refer to previous work by \citet{joachims2017unbiased} and \citet{vardasbi2020trust}.

\citet{oosterhuis2020topkrankings} adapt the \ac{IPS} approach to correct for item-selection bias as well.
They weight clicks according to examination probabilities conditioned on the logging policy, instead of the single displayed ranking on which a click took place.
This results in the \emph{policy-aware} estimator:
\begin{equation}
\begin{split}
\hat{\Delta}_{\text{aware}}(d \mid \bar{\pi}_t, q_t, c_t) &= \frac{c_t(d)}{P(E = 1 \mid \bar{\pi}_t, q_t, d)}
\\
&= 
\frac{c_t(d)}{\sum_{y} \bar{\pi}(y \mid q_t)P(E = 1 \mid y, d, q_t)}.
\end{split}
\end{equation}
This estimator can be used for unbiased optimization under two assumptions.
First, every relevant item must have a non-zero examination probability under the logging policy:
\begin{equation}
\mbox{}\hspace*{-2mm}
\forall t, \forall d \in D_{q_t} \left(P(R=1 \mid, d, q_t) > 0 \rightarrow
P(E = 1 \,|\, \bar{\pi}_t, d, q_t) > 0\right).
\hspace*{-2mm}
\end{equation}
Second, no trust bias is present as described in Eq~\ref{eq:notrustreq}.
Importantly, this first requirement can be met under item-selection bias, since a stochastic ranking policy can always provide every item a non-zero probability of appearing in a top-$k$ ranking.
Thus, even when not all items can be displayed at once, a stochastic policy can provide non-zero examination probabilities to all items.
For a proof of this claim we refer to previous work by \citet{oosterhuis2020topkrankings}.

Lastly, \citet{vardasbi2020trust} prove that \ac{IPS} cannot correct for trust bias.
As an alternative, they introduce an estimator based on affine corrections.
This \emph{affine} estimator penalizes an item displayed at rank $k$ by $\beta_k$ while also reweighting inversely w.r.t.\ $\alpha_k$:
\begin{equation}
\hat{\Delta}_{\text{affine}}(d \mid \bar{y}_t, c_t) =  \frac{c_t(d) - \beta_{d,\bar{y}_t}}{\alpha_{d,\bar{y}_t}}.
\end{equation}
The $\beta$ penalties correct for the number of clicks an item is expected to receive due to its displayed rank, instead of its relevance.
The affine estimator is unbiased under a single assumption, namely that the click probability of every item must be correlated with its relevance in every displayed ranking:
\begin{equation}
\forall t, \forall d \in D_{q_t},  \,
\alpha_{d,\bar{y}_t} \not= 0.
\end{equation}
Thus, while this estimator can correct for position bias and trust bias, it cannot correct for item-selection bias.
For a proof of these claims we refer to previous work by \citet{vardasbi2020trust}.

We note that all of these estimators require knowledge of the position bias ($P(E=1 \mid k)$) or trust bias ($\alpha$ and $\beta$).
A lot of existing work has considered how these values can be inferred accurately~\citep{agarwal2019addressing, wang2018position, fang2019intervention}.
The theory in this paper assumes these values are known.

This concludes our description of existing counterfactual estimators on which our method expands.
To summarize, each of these estimators corrects for position bias, one also corrects for item-selection bias, and another also for trust bias. 
Currently, there is no estimator that corrects for all three forms of bias together.

\section{Related Work}

One of the earliest approaches to \ac{LTR} from clicks was introduced by \citet{joachims2002optimizing}.
It infers pairwise preferences between items from click logs and uses pairwise \ac{LTR} to update an SVM ranking model.
While this approach had some success, in later work \citet{joachims2017unbiased} notes that position bias often incorrectly pushes the pairwise loss to flip the ranking displayed during logging.
To avoid this biased behavior, \citet{joachims2017unbiased} proposed the idea of counterfactual \ac{LTR}, in the spirit of earlier work by \citet{wang2016learning}.
This led to estimators that correct for position bias using \ac{IPS} weighting (see Section~\ref{sec:background:counterfactual}).
This work sparked the field of counterfactual \ac{LTR} which has focused on both capturing interaction biases and optimization methods that can correct for them.
Methods for measuring position bias are based on EM optimization~\citep{wang2018position}, a dual learning objective~\citep{ai2018unbiased}, or randomization~\citep{agarwal2019estimating, fang2019intervention};
for trust bias only an EM-based approach is currently known~\citep{agarwal2019addressing}.
\citet{agarwal2019counterfactual} showed how counterfactual \ac{LTR} can optimize neural networks and \ac{DCG}-like methods through upper-bounding.
\citet{oosterhuis2020topkrankings} introduced an \ac{IPS} estimator that can correct for item-selection bias (see Section~\ref{sec:background:counterfactual}), while also showing that the LambdaLoss framework~\citep{wang2018lambdaloss} can be applied to counterfactual \ac{LTR}.
Lastly, \citet{vardasbi2020trust} proved that \ac{IPS} estimators cannot correct for trust bias and introduced an affine estimator that is capable of doing so (see Section~\ref{sec:background:counterfactual}).
There is currently no known estimator that can correct for position bias, item selection bias, and trust bias simultaneously.

The other paradigm for \ac{LTR} from clicks is online \ac{LTR}~\citep{yue2009interactively}.
The earliest method, \ac{DBGD}, samples variations of a ranking model and compares them using online evaluation~\citep{hofmann2011probabilistic}; if an improvement is recognized the model is updated accordingly.
Most online \ac{LTR} methods have increased the data-efficiency of \ac{DBGD}~\citep{schuth2016mgd, wang2019variance, hofmann2013reusing}; later work found that \ac{DBGD} is not effective at optimizing neural models~\citep{oosterhuis2018differentiable} and often fails to find the optimal linear-model even in ideal scenarios~\citep{oosterhuis2019optimizing}.
To these limitations, alternative approaches for online \ac{LTR} have been proposed. \ac{PDGD} takes a pairwise approach but weights pairs to correct for position bias~\citep{oosterhuis2018differentiable}.
While \ac{PDGD} was found to be very effective and robust to noise~\citep{jagerman2019comparison}, it can be proven that its gradient estimation is affected by position bias, thus we do not consider it to be unbiased.
In contrast, \citet{zhuang2020counterfactual} introduced \ac{COLTR}, which takes the \ac{DBGD} approach but uses a form of counterfactual evaluation to compare candidate models.
Despite making use of counterfactual estimation, \citet{zhuang2020counterfactual} propose the method solely for online \ac{LTR}.

Interestingly, with \ac{COLTR} the line between online and counterfactual \ac{LTR} methods starts to blur.
Recent work by \citet{jagerman2019comparison} applied the original counterfactual approach~\citep{joachims2017unbiased} as an online method and found that it lead to improvements.
Furthermore, \citet{ai2020unbiased} noted that with a small adaptation \ac{PDGD} can be applied to historical data.
Although this means that some existing methods can already be applied both online and counterfactually, no method has been found that is the most reliable choice in both scenarios.

\section{An Estimator Oblivious to Online Interventions}

Before we propose our main contribution, the intervention-aware estimator, we will first introduce an estimator that simultaneously corrects for position bias, item-selection bias, and trust bias, without considering the effects of interventions.
Subsequently, the resulting intervention-oblivious estimator will serve as a method to contrast the intervention-aware estimator with.

Section~\ref{sec:background:counterfactual} described how the policy-aware estimator corrects for item-selection bias by taking into account the behavior of the logging policy used to gather clicks~\citep{oosterhuis2020topkrankings}.
Furthermore, Section~\ref{sec:background:counterfactual} also detailed how the affine estimator corrects for trust bias by applying an affine transformation to individual clicks~\citep{vardasbi2020trust}.
We will now show that a single estimator can correct for both item-selection bias and trust bias simultaneously, by combining the approaches of both these existing estimators.

First we note the probability of a click conditioned on a single logging policy $\pi_t$ can be expressed as:
\begin{align}
P(C=1 |\, d, \pi_t, q)
&= \textstyle \sum_{\bar{y}} \pi_t(\bar{y} \mid q) \mleft( \alpha_{d,\bar{y}} P(R=1 \mid d, q) + \beta_{d,\bar{y}} \mright)
\label{eq:clicktimeoblivious}\\
&= \mathbbm{E}_{\bar{y}}[\alpha_{d} \,|\, \pi_t , q] P(R=1 \,|\, d, q) + \mathbbm{E}_{\bar{y}}[\beta_{d} \,|\, \pi_t, q].
\nonumber
\end{align}
where the expected values of $\alpha$ and $\beta$ conditioned on $\pi_t$ are:
\begin{equation}
\begin{split}
\mathbbm{E}_{\bar{y}}[\alpha_{d} \mid \pi_t, q] = \textstyle \sum_{\bar{y}} \pi_t(\bar{y} \mid q) \alpha_{d,\bar{y}},
\\
\mathbbm{E}_{\bar{y}}[\beta_{d} \mid \pi_t, q] = \textstyle \sum_{\bar{y}} \pi_t(\bar{y} \mid q) \beta_{d,\bar{y}}.
\end{split}
\end{equation}
By reversing Eq.~\ref{eq:clicktimeoblivious} the relevance probability can be obtained from the click probability.
We introduce our \emph{intervention-oblivious estimator}, which applies this transformation to correct for bias:
\begin{equation}
\hat{\Delta}_{\text{IO}}(d \mid q_t, c_t) =  \frac{c_t(d) - \mathbbm{E}_{\bar{y}}[\beta_{d} \mid \pi_t, q_t]}{\mathbbm{E}_{\bar{y}}[\alpha_{d} \mid \pi_t, q_t] }.
\label{eq:timeoblivious}
\end{equation}
The intervention-oblivious estimator brings together the policy-aware and affine estimators: on every click it applies an affine transformation based on the logging policy behavior.
Unlike existing estimators, we can prove that the intervention-oblivious estimator is unbiased w.r.t.\ our assumed click model (Section~\ref{sec:userinteractions}).

\begin{theorem}
\label{theorem:oblivious}
The estimated reward $\hat{\mathcal{R}}$ (Eq.~\ref{eq:estimatedreward}) using the intervention-oblivious estimator (Eq.~\ref{eq:timeoblivious}) is unbiased w.r.t.\ the true reward $\mathcal{R}$ (Eq.~\ref{eq:truereward}) under two assumptions: (1)~our click model (Eq.~\ref{eq:trustbiasmodel}), and (2)~the click probability on every item, conditioned on the logging policies per timestep $\pi_t$, is correlated with relevance:
\begin{equation}
\forall t, \forall d \in D_{q_t}\, \mathbbm{E}_{\bar{y}}[\alpha_{d} \mid \pi_t, q_t] \not= 0.
\label{eq:correlationassumption}
\end{equation}
\end{theorem}

\begin{proof}
Using Eq.~\ref{eq:clicktimeoblivious} and Eq.~\ref{eq:correlationassumption} the relevance probability can be derived from the click probability by:
\begin{equation}
P(R=1 \mid d, q) = \frac{P(C=1 \mid d, \pi_t, q) - \mathbbm{E}_{\bar{y}}[\beta_{d} \mid \pi_t, q]}{\mathbbm{E}_{\bar{y}}[\alpha_{d} \mid \pi_t, q] }.
\label{eq:relderivoblivious}
\end{equation}
Eq.~\ref{eq:relderivoblivious} can be used to show that $\hat{\Delta}_{\text{IO}}$ is an unbiased indicator of relevance:
\begin{align}
\mathbbm{E}_{\bar{y},c}\big[\hat{\Delta}_{\text{IO}}(d \mid q_t, c_t) \mid \pi_t \big] \hspace{-2cm}&
\nonumber \\
&=
\mathbbm{E}_{\bar{y},c}\mleft[\frac{c_t(d) - \mathbbm{E}_{t,\bar{y}}[\beta_{d} \mid \pi_t, q_t]}{\mathbbm{E}_{\bar{y}}[\alpha_{d} \mid \pi_t, q_t] }\mid \pi_t, q_t \mright]
\\
&=  \frac{ \mathbbm{E}_{\bar{y},c}\mleft[ c_t(d)  \mid \pi_t, q_t\mright] - \mathbbm{E}_{\bar{y}}[\beta_{d} \mid \pi_t, q_t]}{\mathbbm{E}_{\bar{y}}[\alpha_{d} \mid \pi_t, q_t] }
\nonumber 
\\
&=  \frac{P(C=1 \mid d, \pi_t, q_t) - \mathbbm{E}_{\bar{y}}[\beta_{d} \mid \pi_t, q_t]}{\mathbbm{E}_{\bar{y}}[\alpha_{d} \mid \pi_t, q_t] } 
\nonumber \\
& = P(R=1 \mid d, q_t).
\nonumber
\label{eq:singleproofoblivious}
\end{align}
Finally, combining Eq.~\ref{eq:truereward} with Eq.~\ref{eq:estimatedreward} and Eq.~\ref{eq:singleproofoblivious} reveals that $\hat{\mathcal{R}}$ based on the intervention-oblivious estimator $\hat{\Delta}_{\text{IO}}$ is unbiased w.r.t.\  $\mathcal{R}$:
\begin{align}
&\mathbbm{E}_{t,q,\bar{y},c}
\mleft[ \hat{\mathcal{R}}(\pi \mid \mathcal{D}) \mright]
\\
& = 
\sum_{q} P(q) \sum_{d \in D_q}  \lambda(d \mid D_{q}, \pi, q) \frac{1}{T} \sum_{t=1}^T \mathbbm{E}_{\bar{y},c}\mleft[ \hat{\Delta}_{\text{IO}}(d \mid c, q) \,|\, \pi_t, q\mright]
\mbox{}\hspace*{-4mm}
\nonumber \\
& = 
\sum_{q} P(q) \sum_{d \in D_q} \lambda(d \mid D_q, \pi, q) P(R = 1 \mid d, q)
= \mathcal{R}(\pi).
\qedhere
\end{align}
\end{proof}

\subsection{Example with an Online Intervention}

\begin{figure}[t]
\centering
{\renewcommand{\arraystretch}{0.2}
\begin{tabular}{c c}
\includegraphics[scale=0.53]{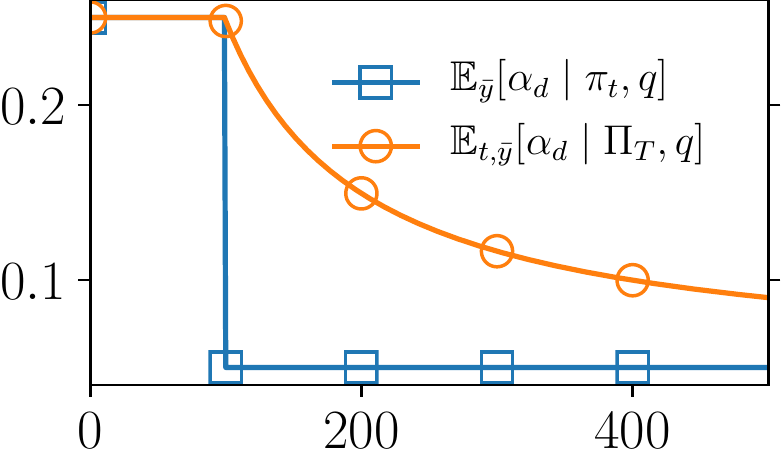} 
&
\includegraphics[scale=0.53]{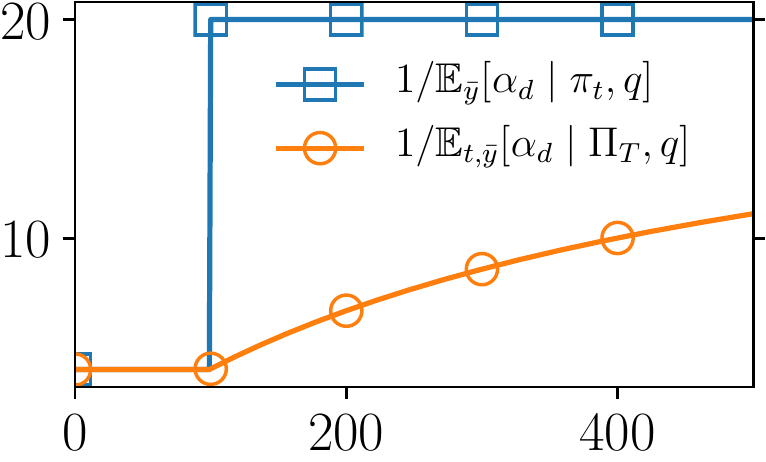}
\\
\small Timestep $T$
&
\small Timestep $T$
\end{tabular}
}
\caption{
Example of an online intervention and the weights used by the intervention-oblivious and intervention-aware estimators for a single item as more data is gathered.
}
\label{fig:intuition}
\end{figure}

Existing estimators for counterfactual \ac{LTR} are designed for a scenario where the logging policy is static: $\forall (\pi_t,\pi_{t'})\in \mathcal{D}, \, \pi_t = \pi_{t'}$.
However, we note that if an online intervention takes place~\citep{jagerman2019comparison}, meaning that the logging policy was updated during the gathering of data: $\exists (\pi_t,\pi_{t'})\in \mathcal{D}, \, \pi_t \not = \pi_{t'}$, the intervention-oblivious estimator is still unbiased.
This was already proven in Theorem~\ref{theorem:oblivious} because its assumptions cover both scenarios where online interventions do and do not take place.

However, the individual corrections of the intervention-oblivious estimator are only based on the single logging policy that was deployed at the timestep of each specific click.
It is completely oblivious to the logging policies applied at different timesteps.
Although this does not lead to bias in its estimation, it does result in unintuitive behavior.
We illustrate this behavior in Figure~\ref{fig:intuition}, here a logging policy that results in $\mathbbm{E}[\alpha_d \mid \pi_t, q] = 0.25$ for an item $d$ is deployed during the first $t \leq 100$ timesteps.
Then an online intervention takes place and the logging policy is updated so that for $t > 100$,  $ \mathbbm{E}[\alpha_d \mid \pi_t, q] = 0.05$.
The intervention-oblivous estimator weights clicks inversely to $\mathbbm{E}[\alpha_d \mid \pi_t]$; so clicks for $t \leq 100$ will be weighted by $1/0.25=4$ and for $t > 100$ by $1/0.05=20$.
Thus, there is a sharp and sudden difference in how clicks are treated before and after $t=100$.
What is unintuitive about this example is that the way clicks are treated after $t=100$ is completely independent of what the situation was before $t=100$.
For instance, consider another item $d'$ where $\forall t, \mathbbm{E}[\alpha_{d'} \mid \pi_t, q] = 0.05$.
If both $d$ and $d'$ are clicked on timestep $t = 101$, these clicks would both be weighted by $20$, despite the fact that $d$ has so far been treated completely different than $d'$.
One would expect that in such a case the click on $d$ should be weighted less, to compensate for the high $\mathbbm{E}[\alpha_d \mid \pi_t, q]$ it had in the first 100 timesteps.
The question is whether such behavior can be incorporated in an estimator without introducing bias.

\section{The Intervention-Aware Estimator}

Our goal for the intervention-aware estimator is to find an estimator whose individual corrections are not only based on single logging policies, but instead consider the entire collection of logging policies used to gather the data $\mathcal{D}$.
Importantly, this estimator should also be unbiased w.r.t.\ position bias, item-selection bias and trust bias.

For ease of notation, we use $\Pi_T$ for the set of policies that gathered the data in $\mathcal{D}$:
$\Pi_T = \{ \pi_1, \pi_2, \ldots, \pi_T\}$.
The probability of a click can be conditioned on this set:
\begin{equation}
\begin{split}
P(C=1 \mid d, \Pi_T, q)  \hspace{-1.8cm} & \\
& = \frac{1}{T}\sum_{t=1}^T \sum_{\bar{y}} \pi_t(\bar{y} \mid q) \mleft( \alpha_{d,\bar{y}} P(R=1 \mid d, q) + \beta_{d,\bar{y}} \mright)
\\
&= \mathbbm{E}_{t,\bar{y}}[\alpha_{d} \mid \Pi_T , q] P(R=1 \mid d, q) + \mathbbm{E}_{t,\bar{y}}[\beta_{d} \mid \Pi_T, q],
\end{split}
\label{eq:clicktimeaware}
\end{equation}
where the expected values of $\alpha$ and $\beta$ conditioned on $\Pi_T$ are:
\begin{equation}
\begin{split}
\mathbbm{E}_{t,\bar{y}}[\alpha_{d} \mid \Pi_T, q] =\frac{1}{T} \textstyle \sum_{t=1}^T \sum_{\bar{y}} \pi_t(\bar{y} \mid q) \alpha_{d,\bar{y}},
\\
\mathbbm{E}_{t,\bar{y}}[\beta_{d} \mid \Pi_T, q] =\frac{1}{T} \textstyle \sum_{t=1}^T \sum_{\bar{y}} \pi_t(\bar{y} \mid q) \beta_{d,\bar{y}}.
\end{split}
\end{equation}
Thus $P(C=1 \mid d, \Pi_T, q)$ gives us the probability of a click given that any policy from $\Pi_T$ could be deployed.
We propose our \emph{intervention-aware estimator} that corrects for bias using the expectations conditioned on $\Pi_T$:
\begin{equation}
\hat{\Delta}_{\text{IA}}(d \mid q_t, c_t) =  \frac{c_t(d) - \mathbbm{E}_{t,\bar{y}}[\beta_{d} \mid \Pi_T, q_t]}{\mathbbm{E}_{t,\bar{y}}[\alpha_{d} \mid \Pi_T, q_t] }.
\label{eq:timeaware}
\end{equation}
The salient difference with the intervention-oblivious estimator is that the expectations are conditioned on $\Pi_T$, all logging policies in $\mathcal{D}$, instead of an individual logging policy $\pi_t$.
While the difference with the intervention-oblivious estimator seems small, our experimental results show that the differences in performance are actually quite sizeable.
Lastly, we note that when no interventions take place the intervention-oblivious estimator and intervention-aware estimators are equivalent.
Because the intervention-aware estimator is the only existing counterfactual \ac{LTR} estimator whose corrections are influenced by online interventions, we consider it to be a step that helps to bridge the gap between counterfactual and online \ac{LTR}.

Before we revisit our online intervention example with our novel intervention-aware estimator, we prove that it is unbiased w.r.t.\ our assumed click model (Section~\ref{sec:userinteractions}).

\begin{theorem}
The estimated reward $\hat{\mathcal{R}}$ (Eq.~\ref{eq:estimatedreward}) using the intervention-aware estimator (Eq.~\ref{eq:timeaware}) is unbiased w.r.t.\ the true reward $\mathcal{R}$ (Eq.~\ref{eq:truereward}) under two assumptions: (1)~our click model (Eq.~\ref{eq:trustbiasmodel}), and (2)~the click probability on every item, conditioned on the set of logging policies $\Pi_T$, is correlated with relevance:
\begin{equation}
\forall q, \forall d \in D_q, \quad  \mathbbm{E}_{t,\bar{y}}[\alpha_{d} \mid \Pi_T, q] \not= 0.
\label{eq:correlationassumption}
\end{equation}
\end{theorem}

\begin{proof}
Using Eq.~\ref{eq:clicktimeaware} and Eq.~\ref{eq:correlationassumption} the relevance probability can be derived from the click probability by:
\begin{equation}
P(R=1 \mid d, q) = \frac{P(C=1 \mid d, \Pi_T, q) - \mathbbm{E}_{t,\bar{y}}[\beta_{d} \mid \Pi_T, q]}{\mathbbm{E}_{t,\bar{y}}[\alpha_{d} \mid \Pi_T, q] }.
\label{eq:relderiv}
\end{equation}
Eq.~\ref{eq:relderiv} can be used to show that $\hat{\Delta}_{\text{IA}}$ is an unbiased indicator of relevance:
\begin{align}
\mathbbm{E}_{t,\bar{y},c}\big[\hat{\Delta}_{\text{IA}}(d \mid q_t, c_t) \mid \Pi_T \big]
\hspace{-2cm}&
\nonumber \\
&=
\mathbbm{E}_{t,\bar{y},c}\mleft[\frac{c_t(d) - \mathbbm{E}_{t,\bar{y}}[\beta_{d} \mid \Pi_T, q_t]}{\mathbbm{E}_{t,\bar{y}}[\alpha_{d} \mid \Pi_T, q_t] }\mid \Pi_T, q_t \mright]
\\
&=  \frac{ \mathbbm{E}_{t,\bar{y},c}\mleft[ c_t(d)  \mid \Pi_T, q_t\mright] - \mathbbm{E}_{t,\bar{y}}[\beta_{d} \mid \Pi_T, q_t]}{\mathbbm{E}_{t,\bar{y}}[\alpha_{d} \mid \Pi_T, q_t] }
\nonumber \\
&=  \frac{P(C=1 \mid d, \Pi_T, q_t) - \mathbbm{E}_{t,\bar{y}}[\beta_{d} \mid \Pi_T, q_t]}{\mathbbm{E}_{t,\bar{y}}[\alpha_{d} \mid \Pi_T, q_t] } 
\nonumber \\
&= P(R=1 \mid d, q_t).
\nonumber
\label{eq:singleproof}
\end{align}
Finally, combining Eq.~\ref{eq:singleproof} with Eq.~\ref{eq:estimatedreward} and Eq.~\ref{eq:truereward} reveals that $\hat{\mathcal{R}}$ based on the intervention-aware estimator $\hat{\Delta}_{\text{IA}}$ is unbiased w.r.t.\  $\mathcal{R}$:
\begin{align}
\mathbbm{E}_{t,q,\bar{y},c}
\mleft[ \hat{\mathcal{R}}(\pi \mid \mathcal{D}) \mright] \hspace{-2.3cm} &
\nonumber \\
& = 
\sum_{q} P(q) \sum_{d \in D_q}  \lambda(d \mid D_{q}, \pi, q) \mathbbm{E}_{t,\bar{y},c}\mleft[ \hat{\Delta}_{\text{IA}}(d \mid c, q) \mid \Pi_T, q\mright]
\nonumber \\
& = 
\sum_{q} P(q) \sum_{d \in D_q} \lambda(d \mid D_q, \pi, q) P(R = 1 \mid d, q) \\
& = \mathcal{R}(\pi).
\qedhere
\end{align}
\end{proof}

\subsection{Online Intervention Example Revisited}

We will now revisit the example in Figure~\ref{fig:intuition}, but this time consider how the intervention-aware estimator treats item $d$.
Unlike the intervention-oblivious estimator, clicks are weighted by $\mathbbm{E}[\alpha_d \mid \Pi_T]$ which means that the exact timestep $t$ of a click does not matter, as long as $t < T$.
Furthermore, the weight of a click can change as the total number of timesteps $T$ increases.
In other words, as more data is gathered, the intervention-aware estimator retroactively updates the weights of all clicks previously gathered.

We see that this behavior avoids the sharp difference in weights of clicks occurring before the intervention $t \leq 100$ and after $t > 100$.
For instance, for a click on $d$ occuring at $t=101$ while $T=400$, results in $\mathbbm{E}[\alpha_d \mid \Pi_T] = 0.1$ and thus a weight of $1/0.1 = 10$.
This is much lower than the intervention-oblivious weight of $1/0.05=20$, because the intervention-aware estimator is also considering the initial period where $\mathbbm{E}[\alpha_{d} \mid \pi_t, q]$ was high.
Thus we see that the intervention-aware estimator has the behavior we intuitively expected: it weights clicks based on how the item was treated throughout all timesteps.
In this example, it leads weights considerably smaller than those used by the intervention-oblivious estimator.
In \ac{IPS} estimators, low propensity weights are known to lead to high variance~\citep{joachims2017unbiased}, thus we may expect that the intervention-aware estimator reduces variance in this example.

\subsection{An Online and Counterfactual Approach}
\label{sec:onlinecounterapproach}

While the intervention-aware estimator takes into account the effect of interventions, it does not prescribe what interventions should take place.
In fact, it will work with any interventions that result in Eq.~\ref{eq:correlationassumption} being true, including the situation where no intervention takes place at all.
For clarity, we will describe the intervention approach we applied during our experiments here.
Algorithm~\ref{alg:online} displays our online/counterfactual approach.
As input it requires a starting policy ($\pi_0$), a choice for $\lambda$, the $\alpha$ and $\beta$ parameters, a set of intervention timesteps ($\Phi$), and the final timestep $T$.

The algorithm starts by initializing an empty set to store the gathered interaction data (Line~\ref{line:initdata}) and initializes the logging policy with the provided starting policy $\pi_0$.
Then for each timestep $i$ in $\Phi$ the dataset is expanded using the current logging policy so that $|\mathcal{D}| = i$ (Line~\ref{line:interventiongather}).
In other words, for $i - |\mathcal{D}|$ timesteps $\pi$ is used to display rankings to user-issued queries, and the resulting interactions are added to $\mathcal{D}$.
Then a policy is optimized using the available data in $\mathcal{D}$ which becomes the new logging policy.
For this optimization, we split the available data in training and validation partitions in order to do early stopping to prevent overfitting.
We use stochastic gradient descent where we use $\pi_0$ as the initial model; this practice is based on the assumption that $\pi_0$ has a better performance than a randomly initialized model.
Thus, during optimization, gradient calculation uses the intervention-aware estimator on the training partition of $\mathcal{D}$, and after each epoch, optimization is stopped if the intervention-aware estimator using the validation partition of $\mathcal{D}$ suspects overfitting.
Each iteration results in an intervention as the resulting policy replaces the logging policy, and thus changes the way future data is logged.
After iterating over $\Phi$ is completed, more data is gathered so that $|\mathcal{D}| = T$ and optimization is performed once more.
The final policy is the end result of the procedure.

We note that, depending on $\Phi$, our approach can be either online, counterfactual, or somewhere in between.
If $\Phi = \emptyset$ the approach is fully counterfactual since all data is gathered using the static $\pi_0$.
Conversely, if $\Phi = \{1,2,3,\ldots,T\}$ it is fully online since at every timestep the logging policy is updated.
In practice, we expect a fully online procedure to be infeasible as it is computationally expensive and user queries may be issued faster than optimization can be performed.
In our experiments we will investigate the effect of the number of interventions on the approach's performance.

\section{Experimental Setup}

Our experiments aim to answer the following research questions:
\begin{enumerate}[align=left, label={\bf RQ\arabic*},leftmargin=*]
\item Does the intervention-aware estimator lead to higher performance than existing counterfactual \ac{LTR} estimators when online interventions take place?
\item Does the intervention-aware estimator lead to performance comparable with existing online \ac{LTR} methods?
\end{enumerate}
We use the semi-synthetic experimental setup that is common in existing work on both online \ac{LTR}~\citep{oosterhuis2018differentiable, oosterhuis2019optimizing, hofmann2013reusing, zhuang2020counterfactual} and counterfactual \ac{LTR}~\citep{vardasbi2020trust, ovaisi2020correcting, joachims2017unbiased}.
In this setup, queries and documents are sampled from a dataset based on commercial search logs, while user interactions and rankings are simulated using probabilistic click models.
The advantage of this setup is that it allows us to investigate the effects of online interventions on a large scale while also being easy to reproduce by researchers without access to live ranking systems.

\begin{algorithm}[t]
\caption{Our Online/Counterfactual \ac{LTR} Approach} 
\label{alg:online}
\begin{algorithmic}[1]
\STATE \textbf{Input}: Starting policy: $\pi_0$; Metric weight function: $\lambda$;\\
\qquad\quad
Inferred bias parameters: $\alpha$ and $\beta$;\\
\qquad\quad
Interventions steps: $\Phi$; End-time: $T$.
\STATE $\mathcal{D} \leftarrow \{\}$ \hfill \textit{\small // initialize data container} \label{line:initdata}
\STATE $\pi \gets \pi_0$ \hfill \textit{\small // initialize logging policy}  \label{line:initpolicy}
\FOR{$i \in \Phi$}
\STATE $\mathcal{D} \leftarrow \mathcal{D} \cup \text{gather}(\pi, i-|\mathcal{D}|)$  \hfill \textit{\small // observe $i-|\mathcal{D}|$ timesteps}  \label{line:interventiongather}
\STATE $\pi \leftarrow \text{optimize}(\mathcal{D}, \alpha, \beta, \pi_0)$ \hfill \textit{\small // optimize based on available data} \label{line:interventionoptimize}
\ENDFOR
\STATE $\mathcal{D} \leftarrow \mathcal{D} \cup \text{gather}(\pi, T-|\mathcal{D}|)$  \hfill \textit{\small // expand data to $T$} \label{line:finalgather}
\STATE $\pi \leftarrow \text{optimize}(\mathcal{D}, \alpha, \beta, \pi_0)$ \hfill \textit{\small // optimize based on final data} \label{line:finaloptimize}
\RETURN $\pi$ 
\end{algorithmic}
\end{algorithm}

We use the publicly-available Yahoo Webscope dataset~\citep{Chapelle2011}, which consists of \numprint{29921} queries with, on average, 24 documents preselected per query.
Query-document pairs are represented by 700 features and five-grade relevance annotations ranging from not relevant (0) to perfectly relevant (4). The queries are divided into training, validation and test partitions.

At each timestep, we simulate a user-issued query by uniformly sampling from the training and validation partitions.
Subsequently, the preselected documents are ranked according to the logging policy, and user interactions are simulated on the top-5 of the ranking using a probabilistic click model.
We apply Eq.~\ref{eq:clickmodelused} with $\alpha = [0.35, 0.53, 0.55, 0.54, 0.52]$ and $\beta = [0.65, 0.26, 0.15, 0.11, 0.08]$; the relevance probabilities are based on the annotations from the dataset: $P(R = 1 \mid d,q) = 0.25 \cdot \text{relevance\_label}(d,q)$.
The values of $\alpha$ and $\beta$ were chosen based on those reported by \citet{agarwal2019addressing} who inferred them from real-world user behavior.
In doing so, we aim to emulate a setting where realistic levels of position bias, item-selection bias, and trust bias are present.

All counterfactual methods use the approach described in Section~\ref{sec:onlinecounterapproach}.
To simulate a production ranker policy, we use supervised \ac{LTR} to train a ranking model on 1\% of the training partition~\citep{joachims2017unbiased}.
The resulting production ranker has much better performance than a randomly initialized model, yet still leaves room for improvement.
We use the production ranker as the initial logging policy.
The size of $\Phi$ (the intervention timesteps) varies per run, and the timesteps in $\Phi$ are evenly spread on an exponential scale.
All ranking models are neural networks with two hidden layers, each containing 32 hidden units with sigmoid activations.
Gradients are calculated using a Monte-Carlo method following \citet{oosterhuis2020taking}.
All policies apply a softmax to the document scores produced by the ranking models to obtain a probability distribution over documents.
Clipping is only applied on the training clicks, denominators of any estimator are clipped by $10/\sqrt{|\mathcal{D}|}$ to reduce variance.
Early stopping is applied based on counterfactual estimates of the loss using (unclipped) validation clicks.

The following methods are compared:
\begin{enumerate*}[label=(\roman*)]
\item The intervention-aware estimator.
\item The intervention-oblivious estimator.
\item The policy-aware estimator~\citep{oosterhuis2020topkrankings}.
\item The affine estimator~\citep{vardasbi2020trust}.
\item \ac{PDGD}~\citep{oosterhuis2018differentiable}, we apply \ac{PDGD} both online and as a counterfactual method.
As noted by \citet{ai2020unbiased}, this can be done by separating the logging models from the learned model and, basing the debiasing weights on the logging function.
\item Biased \ac{PDGD}, identical to \ac{PDGD} except that we do not apply the debiasing weights.
\item \ac{COLTR}~\citep{zhuang2020counterfactual}.
\end{enumerate*}
We compute the \ac{NDCG} of both the logging policy and of a policy trained on all available data.
Every reported result is the average of 20 independent runs, figures plot the mean, shaded areas indicate 90\% confidence bounds.
To facilitate reproducibility, our implementation will be made publicly available.

\begin{figure}[t]
\centering
{\renewcommand{\arraystretch}{0.2}
\begin{tabular}{r c}
\rotatebox[origin=lt]{90}{\footnotesize \hspace{4.3em} NDCG} &
\includegraphics[scale=0.44]{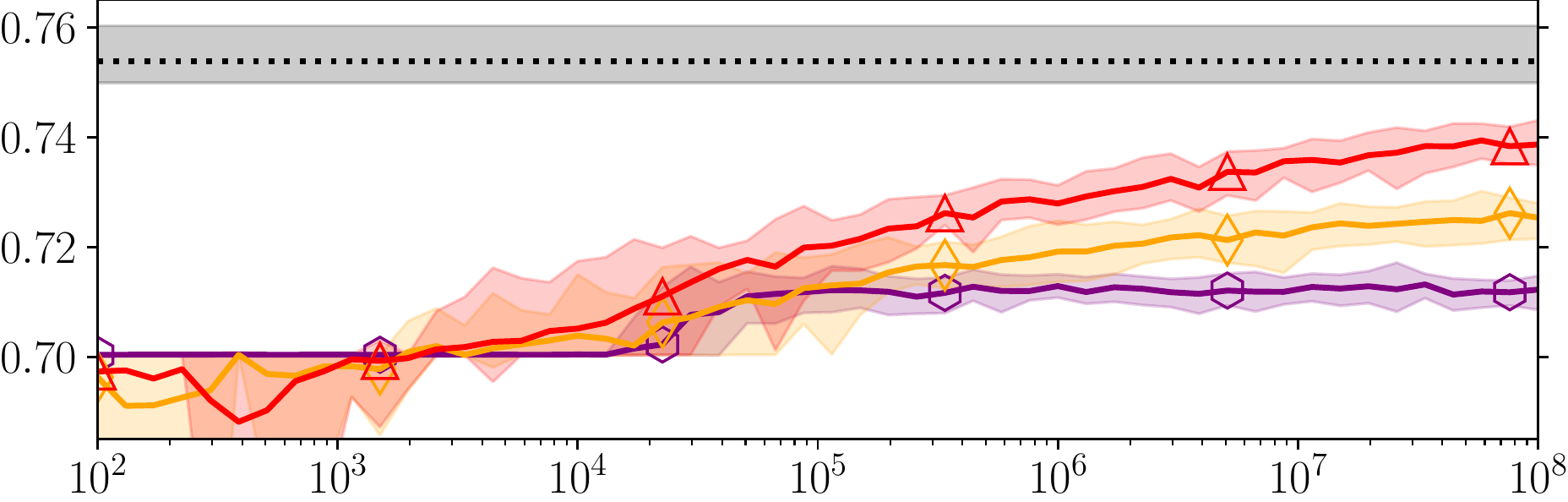}
\\
\rotatebox[origin=lt]{90}{\footnotesize \hspace{4.3em} NDCG} &
\includegraphics[scale=0.44]{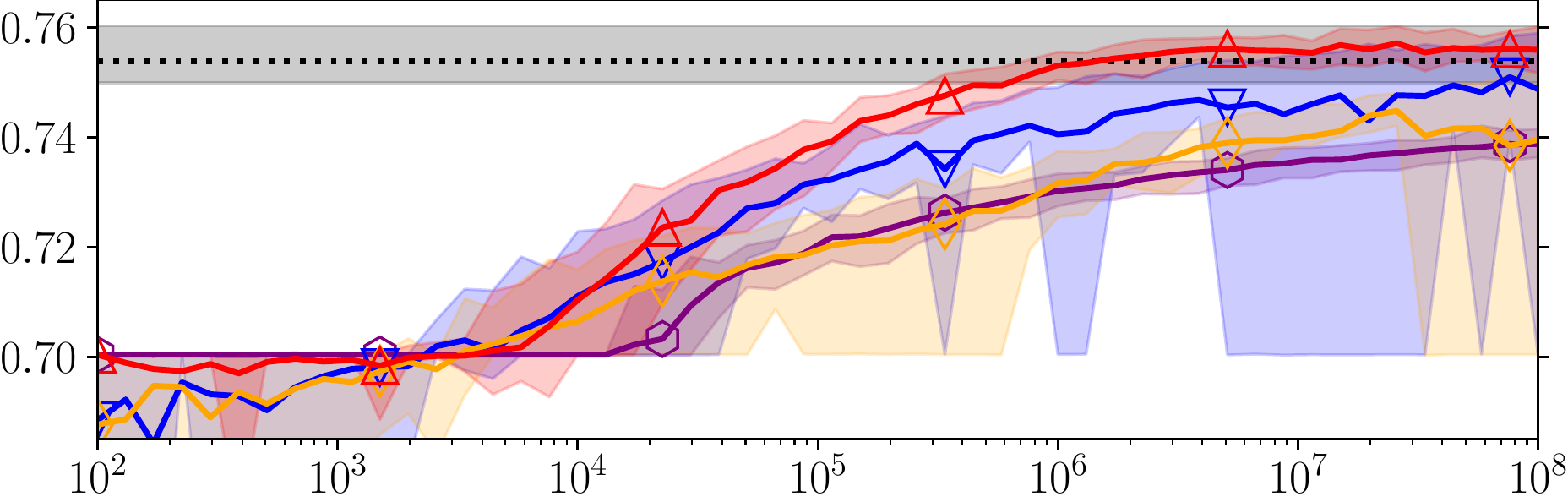}
\\
& \footnotesize \hspace{1em} Number of Logged Queries
\end{tabular}
}
\includegraphics[scale=0.35]{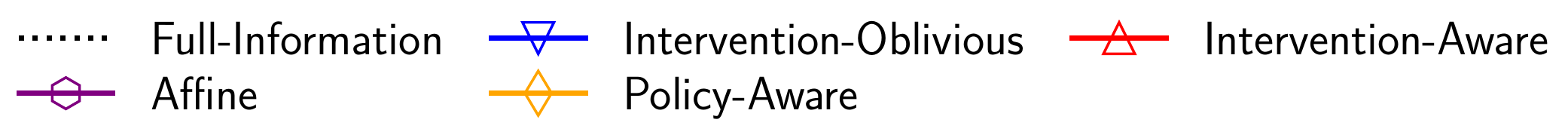}
\caption{
Comparison of counterfactual \ac{LTR} estimators. Top: Counterfactual runs (no interventions); Bottom: Online runs (50 interventions).
Results based on an average of 20 runs, shaded area indicates the 90\% confidence bounds.
}
\label{fig:estimators}
\end{figure}

\section{Results and Discussion}

\subsection{Comparison with Counterfactual \ac{LTR}}

To answer the first research question: \emph{whether the intervention-aware estimator leads to higher performance than existing counterfactual LTR estimators when online interventions take place},
we consider Figure~\ref{fig:estimators} which displays the performance of \ac{LTR} using different counterfactual estimators.

First we consider the top of Figure~\ref{fig:estimators} which displays performance in the counterfactual setting where the logging policy is static.\footnote{Since under a static logging policy the intervention-aware and the intervention-oblivious estimators are equivalent, our conclusions apply to both in this setting.}
We clearly see that the affine estimator converges at a suboptimal point of convergence, a strong indication of bias.
The most probable cause is that the affine estimator is heavily affected by the presence of item-selection bias.
In contrast, neither the policy-aware estimator nor the intervention-aware estimator have converged after $10^8$ queries.
However, very clearly the intervention-aware estimator quickly reaches a higher performance.
While the theory guarantees that it will converge at the optimal performance, we were unable to observe the number of queries it requires to do so.
From the result in the counterfactual setting, we conclude that by correcting for position-bias, trust-bias, and item-selection bias the intervention-aware estimator already performs better without online interventions.

Second, we turn to the bottom of Figure~\ref{fig:estimators} which considers the online setting where the estimators perform 50 online interventions during logging.
We see that online interventions have a positive effect on all estimators; leading to a higher performance for the affine and policy-aware estimators as well.
However, interventions also introduce an enormous amount of variance for the policy-aware and intervention-oblivious estimators.
In stark contrast, the amount of variance of the intervention-aware estimator hardly increases while it learns much faster than the other estimators.

Thus we answer the first research question positively:
the inter\-vention-aware estimator leads to higher performance than existing estimators, moreover, its data-efficiency becomes even greater when online interventions take place.

\begin{figure}[t]
\centering
{\renewcommand{\arraystretch}{0.2}
\begin{tabular}{r c}
\rotatebox[origin=lt]{90}{\footnotesize \hspace{4.3em} NDCG} &
\includegraphics[scale=0.44]{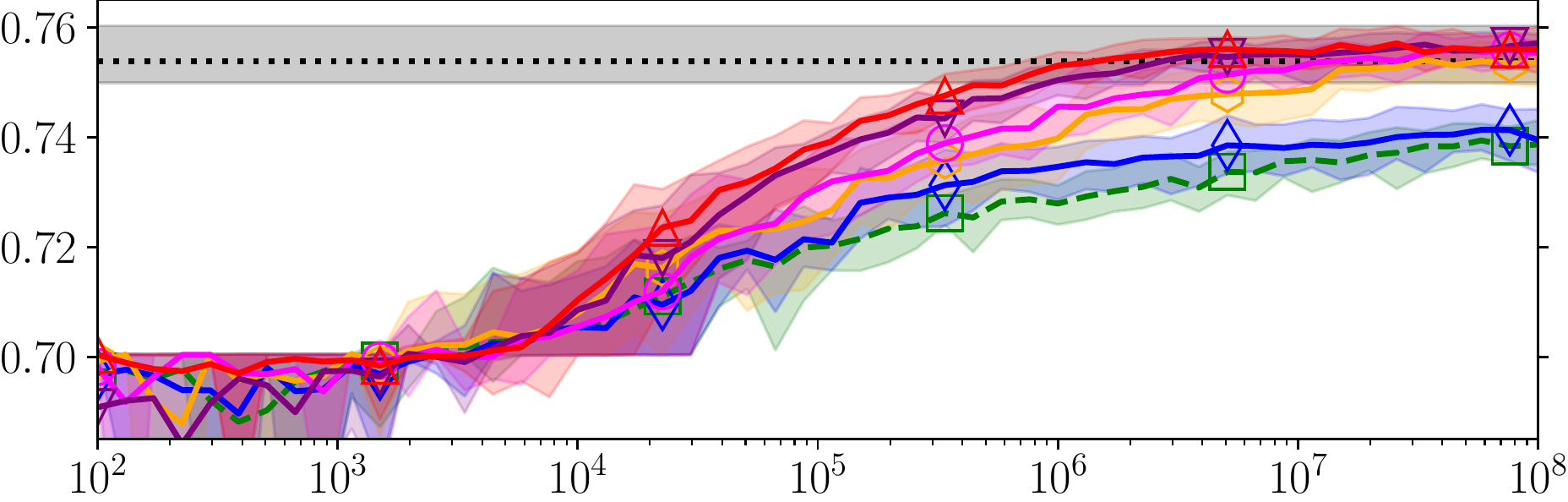}\\
\rotatebox[origin=lt]{90}{\footnotesize \hspace{0.6em} Logging Policy NDCG} &
\includegraphics[scale=0.44]{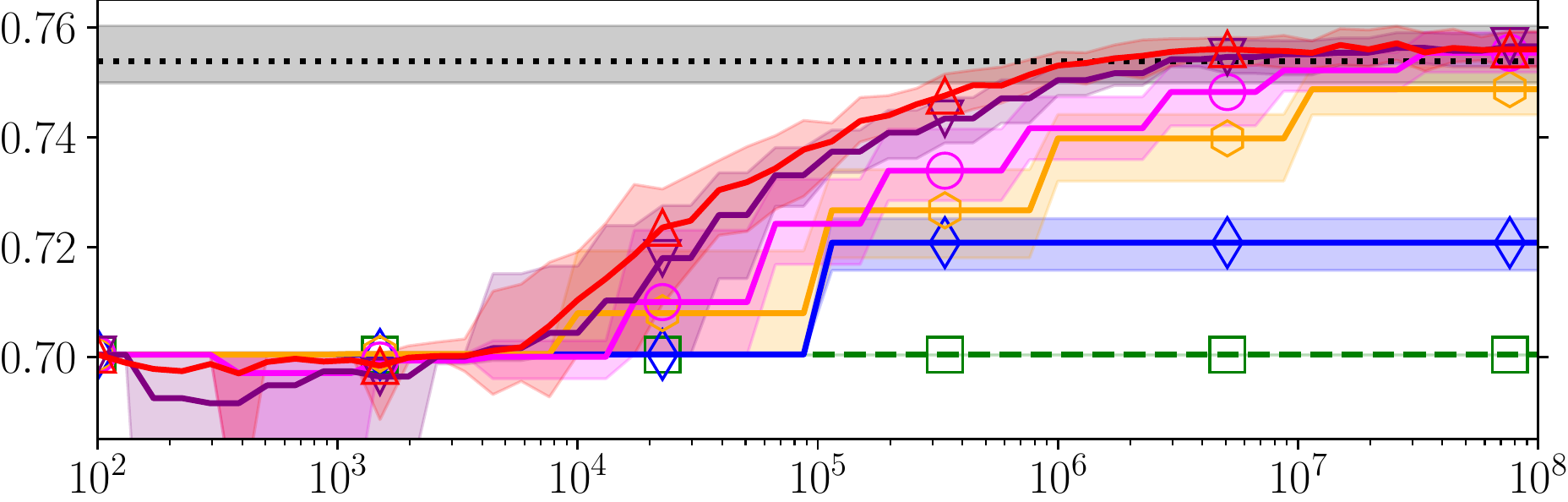}
\\
& \footnotesize \hspace{1em} Number of Logged Queries
\end{tabular}
}
\includegraphics[scale=0.365]{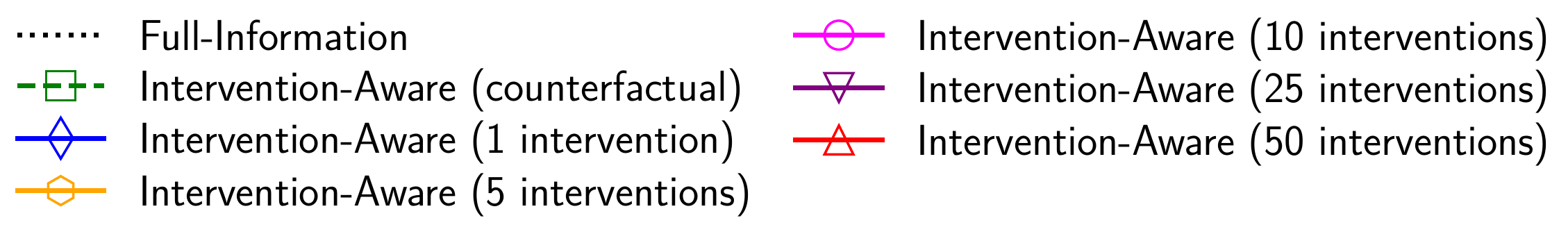}
\caption{
Effect of online interventions on \ac{LTR} with the intervention-aware estimator.
Results based on an average of 20 runs, shaded area indicates the 90\% confidence bounds.
}
\label{fig:interventions}
\end{figure}

\subsection{Effect of Interventions}

To better understand how much the intervention-aware estimator benefits from online interventions, we compared its performance under varying numbers of interventions in Figure~\ref{fig:interventions}.
It shows both the performance of the resulting model when training from the logged data (top), as the performance of the logging policy which reveals when interventions take place (bottom).
When comparing both graphs, we see that interventions lead to noticeable immediate improvements in data-efficiency.
For instance, when only 5 interventions take place the intervention-aware estimator needs more than 20 times the amount of data to reach optimal performance as with 50 interventions. 
Despite these speedups there are no large increases in variance.
From these observations, we conclude that the intervention-aware estimator can effectively and reliably utilize the effect of online interventions for optimization, leading to enormous increases in data-efficiency.

\begin{figure}[t]
\centering
{\renewcommand{\arraystretch}{0.2}
\begin{tabular}{r c}
\rotatebox[origin=lt]{90}{\footnotesize \hspace{4.3em} NDCG} &
\includegraphics[scale=0.44]{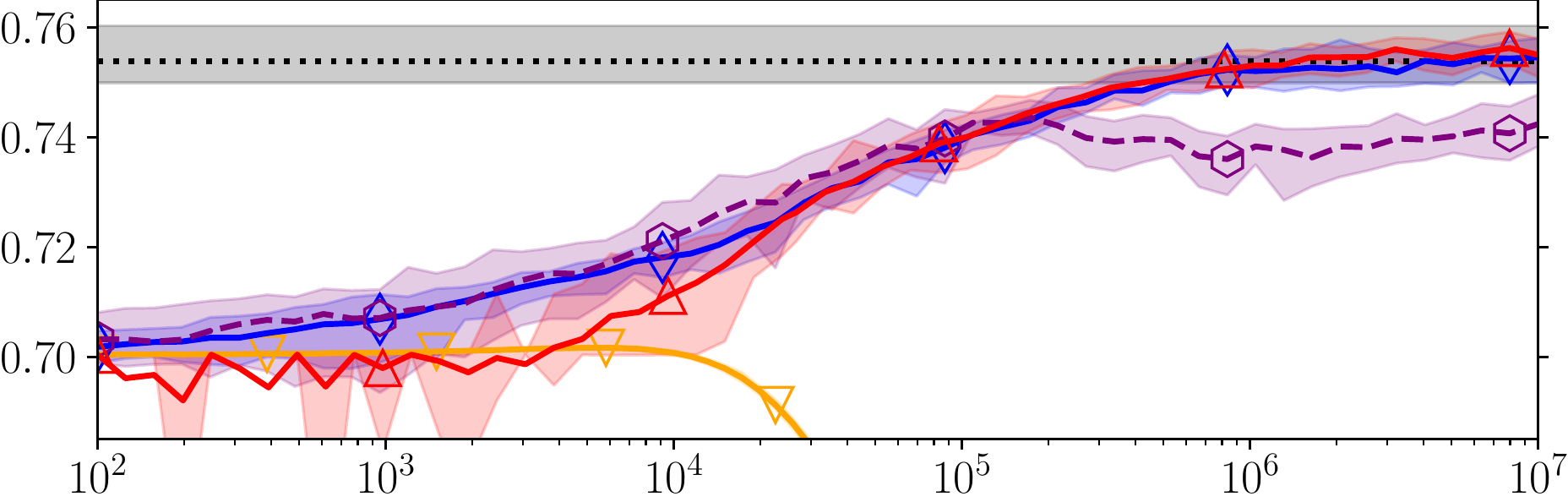}\\
\rotatebox[origin=lt]{90}{\footnotesize \hspace{0.6em} Logging Policy NDCG} &
\includegraphics[scale=0.44]{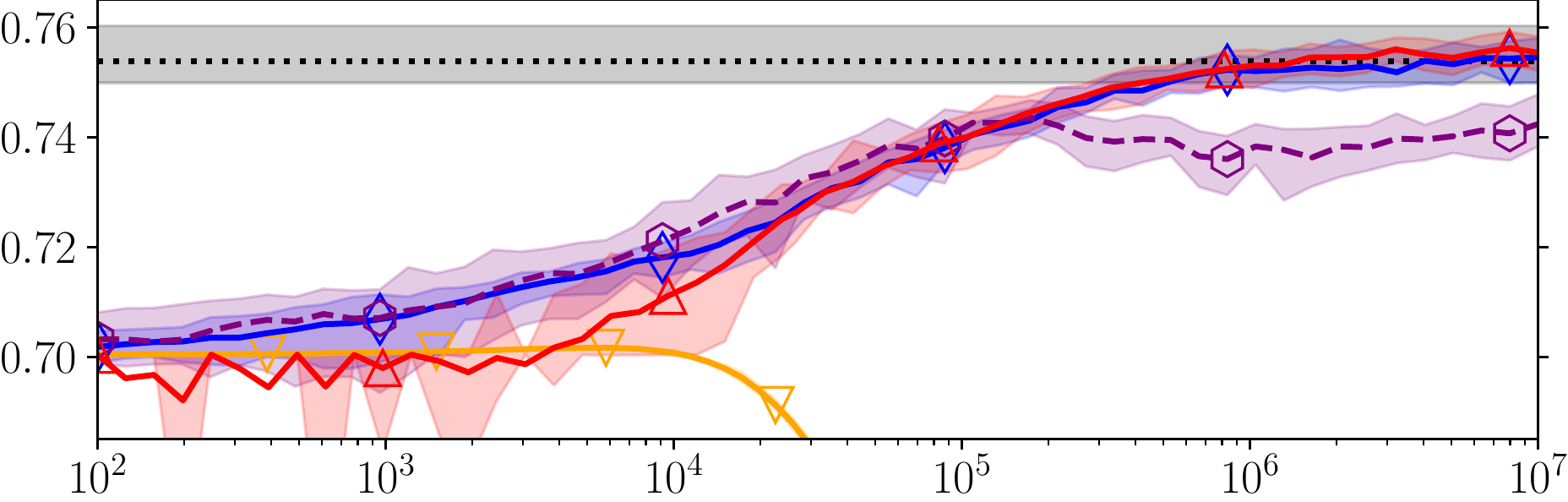}
\\
& \footnotesize \hspace{1em} Number of Logged Queries
\end{tabular}
}
\includegraphics[scale=0.38]{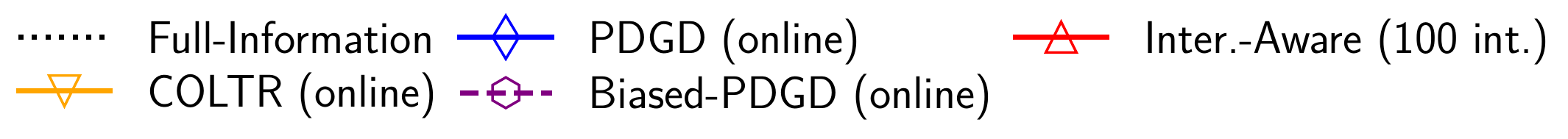} 
\caption{
Comparison with online \ac{LTR} methods.
Results based on an average of 20 runs, shaded area indicates the 90\% confidence bounds.
}
\label{fig:pdgdcomp}
\end{figure}

\subsection{Comparison with Online \ac{LTR}}

In order to answer the second research question: \emph{whether the inter\-vention-aware estimator leads to performance comparable with existing online \ac{LTR} methods},
we consider Figure~\ref{fig:pdgdcomp} which displays the performance of two online \ac{LTR} methods: \ac{PDGD} and \ac{COLTR} and the intervention-aware estimator with 100 online interventions.

First, we notice that \ac{COLTR} is unable to outperform its initial policy, moreover, we see its performance drop as the number of iterations increase.
We were unable to find hyper-parameters for \ac{COLTR} where this did not occur.
It seems likely that \ac{COLTR} is unable to deal with trust-bias, thus causing this poor performance.
However, we note that \citet{zhuang2020counterfactual} already show \ac{COLTR} performs poorly when no bias or noise is present, suggesting that it is perhaps an unstable method overall.

Second, we see that the difference between \ac{PDGD} and the inter\-vention-aware estimator becomes negligible after $2\cdot10^4$ queries.
Despite \ac{PDGD} running fully online, and the intervention-aware estimator only performing 100 interventions in total.
We do note that \ac{PDGD} initially outperforms the intervention-aware estimator, thus it appears that \ac{PDGD} works better with low numbers of interactions.
Additionally, we should also consider the difference in overhead: while \ac{PDGD} requires an infrastructure that allows for fully online learning, the intervention-aware estimator only requires 100 moments of intervention, yet has comparable performance after a short initial period.
By comparing Figure~\ref{fig:pdgdcomp} to Figure~\ref{fig:estimators}, we see that the intervention-aware estimator is the first counterfactual \ac{LTR} estimator that leads to stable performance while being comparably efficient with online \ac{LTR} methods.

Thus we answer the second research question positively:
besides an initial period of lower performance, the intervention-aware estimator has comparable performance to online \ac{LTR} and only requires 100 online interventions to do so.
To the best of our knowledge, it is the first counterfactual \ac{LTR} method that can achieve this feat.

\subsection{Understanding the Effectiveness of \ac{PDGD}}
Now that we concluded that the intervention-aware estimator reaches performance comparable to \ac{PDGD} when enough online interventions take place, the opposite question seems equally interesting:
\emph{Does \ac{PDGD} applied counterfactually provide performance comparable to existing counterfactual \ac{LTR} methods?}

To answer this question, we ran \ac{PDGD} in a counterfactual way following \citet{ai2020unbiased}, both fully counterfactual and with only 100 interventions.
The results of these runs are displayed in Figure~\ref{fig:pdgdcount}.
Quite surprisingly, the performance of \ac{PDGD} ran counterfactually and with 100 interventions, reaches much higher performance than the intervention-aware estimator without interventions.
However, after a peak in performance around $10^6$ queries, the \ac{PDGD} performance starts to drop.
This drop cannot be attributed to overfitting, since online \ac{PDGD} does not show the same behavior.
Therefore, we must conclude that \ac{PDGD} is biased when not ran fully online.
This conclusion does not contradict the existing theory, since \citet{oosterhuis2018differentiable} only proved it is unbiased w.r.t.\ \emph{pairwise} preferences.
In other words, \ac{PDGD} is not proven to unbiasedly optimize a ranking metric, thus also not proven to converge on the optimal model.
This drop is particularly unsettling because \ac{PDGD} is a continuous learning algorithm: there is no known early stopping method for \ac{PDGD}.
Yet these results show there is a great risk in running \ac{PDGD} for too many iterations if it is not applied fully online.
To answer our \ac{PDGD} question: although \ac{PDGD} reaches high performance when run counterfactually and appears to have great data-efficiency initially, it appears to converge at a suboptimal biased model.
Thus we cannot conclude \ac{PDGD} is a reliable method for counterfactual \ac{LTR}.

To better understand \ac{PDGD}, we removed its debiasing weights resulting in the performance shown in Figure~\ref{fig:pdgdcomp} (Biased-PDGD).
Clearly, \ac{PDGD} needs these weights to reach optimal performance. 
Similarly, from Figure~\ref{fig:pdgdcount} we see it also needs to be run fully online.
This makes the choice between the intervention-aware estimator and \ac{PDGD} complicated: on the one hand, \ac{PDGD} does not require us to know the $\alpha$ and $\beta$ parameters, unlike the intervention-aware estimator; furthermore, \ac{PDGD} has better initial data-efficiency even when not run fully online.
On the other hand, there are no theoretical guarantees for the convergence of \ac{PDGD}, and we have observed that not running it fully online can lead to large drops in performance.
It seems the choice ultimately depends on what guarantees a practitioner prefers.

\begin{figure}[t]
\centering
{\renewcommand{\arraystretch}{0.2}
\begin{tabular}{r c}
\rotatebox[origin=lt]{90}{\footnotesize \hspace{4.3em} NDCG} &
\includegraphics[scale=0.44]{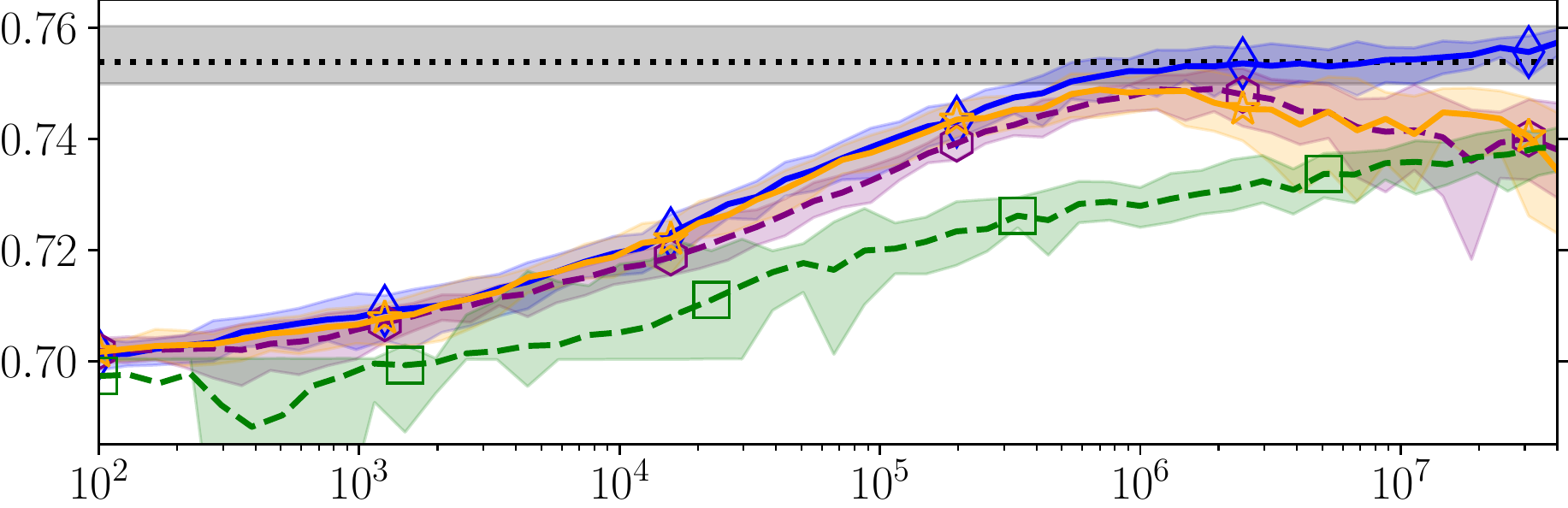}\\
\rotatebox[origin=lt]{90}{\footnotesize \hspace{0.6em} Logging Policy NDCG} &
\includegraphics[scale=0.44]{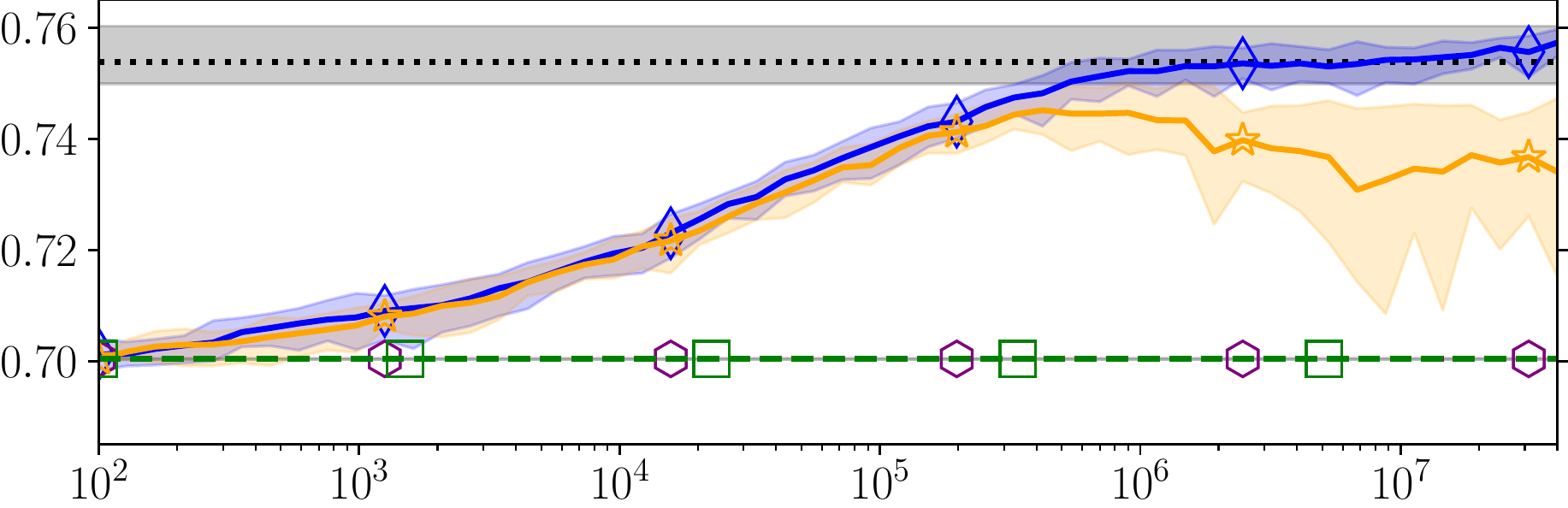}
\\
& \footnotesize \hspace{1em} Number of Logged Queries
\end{tabular}
}
\includegraphics[scale=0.35]{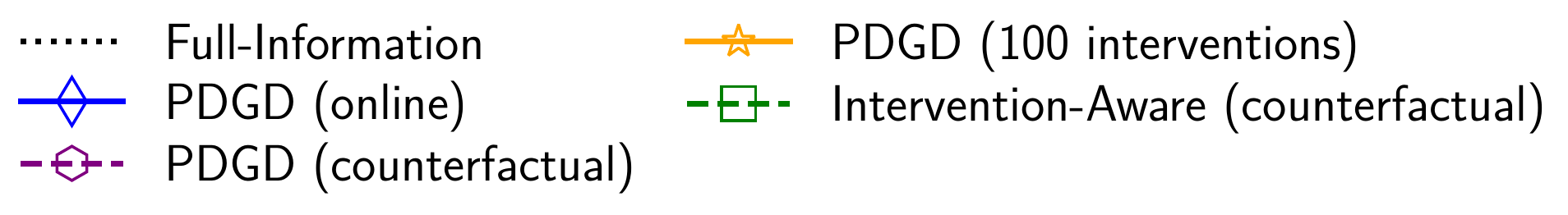} 
\caption{
Effect of online interventions on \ac{PDGD}.
Results based on an average of 20 runs, shaded area indicates the 90\% confidence bounds.
}
\label{fig:pdgdcount}
\end{figure}

\section{Conclusion}

In this paper, we have introduced an intervention-aware estimator: an extension of existing counterfactual approaches that corrects for position-bias, trust-bias, and item-selection bias, while also considering the effect of online interventions.
Our results show that the intervention-aware estimator outperforms existing counterfactual \ac{LTR} estimators, and greatly benefits from online interventions in terms of data-efficiency.
With only 100 interventions it is able to reach performance comparable to state-of-the-art online \ac{LTR} methods.

With the introduction of the intervention-aware estimator, we hope to further unify the fields of online \ac{LTR} and counterfactual \ac{LTR} as it appears to be the most reliable method for both settings.
Future work could investigate what kind of interventions work best for the intervention-aware estimator.

\begin{acks}
This research was partially supported by the Netherlands Organisation for Scientific Research (NWO) under project nr 612.001.551 and by the Innovation Center for AI (ICAI).
All content represents the opinion of the authors, which is not necessarily shared or endorsed by their respective employers and/or sponsors.
\end{acks}

\section*{Code and data}
To facilitate the reproducibility of the reported results, this work only made use of publicly available data and our experimental implementation is publicly available at \url{https://github.com/HarrieO/2021wsdm-unifying-LTR}.

\bibliographystyle{ACM-Reference-Format}
\bibliography{thesis}

\appendix

\end{document}